\documentclass[a4paper,UKenglish,cleveref,autoref,thm-restate]{lipics-v2021}

\usepackage{mathtools}
\usepackage{stmaryrd}
\usepackage{bbold}
\usepackage{todonotes}
\usepackage{xspace}
\usepackage{quiver}
\usepackage{mathpartir}
\usepackage[new]{old-arrows}
\usepackage{wasysym}

\usepackage[autostyle]{csquotes}

\usepackage[outline]{contour}
\contourlength{0.2pt}

\hideLIPIcs

\title{Type Theory With Erasure}

\author{Constantine Theocharis}{University of St Andrews}{kt81@st-andrews.ac.uk}{https://orcid.org/0000-0002-9734-367X}{}
\author{Edwin Brady}{University of St Andrews}{ecb10@st-andrews.ac.uk}{https://orcid.org/0009-0001-0198-2750}{}

\authorrunning{C. Theocharis and E. Brady}

\Copyright{Constantine Theocharis and Edwin Brady}

\ccsdesc[500]{Theory of computation~Type theory}

\keywords{Type theory, erasure, dependent types, compilation, synthetic phase distinction, higher-order abstract syntax, logical frameworks}

\EventEditors{Frank Pfenning}
\EventNoEds{1}
\EventLongTitle{11th International Conference on Formal Structures for Computation and Deduction (FSCD 2026)}
\EventShortTitle{FSCD 2026}
\EventAcronym{FSCD}
\EventYear{2026}
\EventDate{July 20--23, 2026}
\EventLocation{Lisbon, Portugal}
\EventLogo{}
\SeriesVolume{378}
\ArticleNo{9}

\supplementdetails{Implementation}{https://github.com/kontheocharis/erasure-impl}
\supplementdetails{Formalisation}{https://github.com/kontheocharis/erasure-agda}

\acknowledgements{We thank András Kovács, Szumi Xie, Bhakti Shah, Na\"im Camille Favier and the anonymous reviewers for valuable discussions and feedback.}

\nolinenumbers

\newcommand{\s}[1]{{\ensuremath{\textsf{\upshape #1}}}\xspace}

\newcommand{\cl}[1]{{\ensuremath{\mathcal{#1}}}\xspace}
\newcommand{\un}[1]{{\ensuremath{\texttt{\upshape #1}}}\xspace}

\newcommand{\re}[1]{{\ensuremath{\mathbf{#1}}}\xspace}
\newcommand{\std}[1]{\ensuremath{\mathbb{#1}}\xspace}

\newcommand{\mth}[1]{\ensuremath{\textbf{\upshape #1}}\xspace}

\newcommand{\Set}{\mth{Set}}

\newcommand{\Psh}[1]{\ensuremath{\re{Psh}({#1})}\xspace}

\newcommand{\yoneda}{\s{y}}

\newcommand{\col}[1]{\stackrel{#1}{:}}

\newcommand{\Con}{\s{Con}}
\newcommand{\Sub}{\s{Sub}}

\newcommand{\Ty}{\s{Ty}}

\newcommand{\Tm}{\s{Tm}}

\newcommand{\univ}{\cl{U}}
\newcommand{\univR}{\cl{U}_{\mathsf{R}}}

\newcommand{\Unit}{\mathbb{1}}

\newcommand{\llb}[1]{\llbracket{#1}\rrbracket}

\newcommand{\ZTT}{\ensuremath{\s{TT}_0}\xspace}

\newcommand{\TT}{\ensuremath{\s{TT}}\xspace}

\newcommand{\lR}{_\s{R}}

\newcommand{\amazing}[2]{\sqrt[\uproot{4}{#1}]{#2}}
\newcommand{\emb}[1]{\ulcorner{#1}\urcorner}
\newcommand{\lemb}[1]{\llcorner{#1}\lrcorner}

\def\arrowscale{0.8}
\def\arrowshift{0.9pt}
\newcommand{\arrowwrap}[2]{%
  \text{\raisebox{\arrowshift}{\scalebox{\arrowscale}{{\contour[32]{black}{$#1$}}$_{#2}\,$}}}%
}
\newcommand{\arrowwrapalt}[2]{%
  \text{\raisebox{\arrowshift}{\scalebox{\arrowscale}{{\contour[32]{black}{$#1$}}$_{#2}$}}}%
}
\newcommand{\uptm}[1][]{\arrowwrap{\uparrow}{#1}}
\newcommand{\downtm}[1][]{\arrowwrap{\downarrow}{#1}}

\newcommand{\upsub}{\arrowwrapalt{\upuparrows}{}}

\newcommand{\agdaicon}[1]{{\upshape\bfseries{\href{#1}{\raisebox{-0.2\fontcharht\font`I}{\includegraphics[height=1.3\fontcharht\font`I]{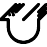}}}}}}

\begin{document}

\maketitle

\begin{abstract}
Erasure enriches type theory with a distinction between runtime relevant and
irrelevant data, allowing the compilation step to safely erase the latter.
Versions of this feature are implemented by many systems, including Agda, Idris,
and Rocq. We present a structural version of type theory with erasure,
formulated as a second-order generalised algebraic theory (SOGAT). Erasure is
encoded as a phase distinction between runtime and erased terms, in the form of
a proposition that can appear in a context. This formulation has several
advantages: it has models based on categories with families, is compatible with
other structural features such as staging, and provides a better guideline for
implementation. Through the model theory of SOGATs, we study the semantics of
type theory with erasure in families of sets, which generalises to any Grothendieck
topos equipped with a tiny proposition. We establish conservativity over
Martin-L\"of type theory (MLTT) in both phases. For code extraction, we construct a
presheaf model that produces untyped lambda calculus programs and prove its
correctness through gluing. Our results are formalised in Agda and we provide a
toy elaborator implementation.
\end{abstract}

\section{Introduction}\label{sec:intro}

When types depend on values, it becomes unclear which parts of a program are
needed at runtime. One can perform whole-program analysis steps to heuristically
detect and remove the majority of data that is not computationally relevant
\cite{Brady2004-ay,Tejiscak2020-rg}. However, the unbounded abstraction
capabilities of dependent types make such techniques brittle and unpredictable
for more complicated examples. What seems to have become the standard way to
erase runtime-irrelevant data is to enrich the underlying type theory with a
separate sort for terms to be erased by compilation. One popular approach
is quantitative type theory (QTT) \cite{McBride2016-oa,Atkey2018-pj}. Parameterised by a
\emph{quantitative semiring}, it encapsulates not only erasure but also
various forms of substructural variable usage, including linearity. When
instantiated to the ordered semiring $\{0 < \omega\}$, the resulting theory
could be described as `Martin-L\"of type theory (MLTT) with erasure'. This approach is
implemented by Agda \cite{agda-irrelevance} and Idris \cite{idris-quantities}.
Besides this, there is a mechanism of `ghost sorts' \cite{Winterhalter2024-pw}
planned for Rocq \cite{Winterhalter2025-tt} that achieves
similar goals.

We contribute yet another approach to erasure, in the style of algebraic,
reduction-free `type theory in type theory' \cite{Altenkirch2016-zc}. The main
difference over previous work is that our formulation is fully structural,
corresponding to a well-studied class of type theories identified by
Uemura~\cite{Uemura2021-jq}. As a result, we can rely on certain aspects of its
metatheory and its implementation that are already understood in the general
setting. We present erasure as a modular \emph{second-order generalised
algebraic theory} (SOGAT) extending MLTT, compatible with cumulative universes
and inductive families.

From this formulation, we derive both syntactic and semantic results. On the
syntactic side, we establish conservativity over MLTT
(\cref{thm:erased-cons,thm:rt-cons}) and the independence of runtime data from
erased data (\cref{thm:need-nothing}). The former ensures that erasure does not
alter the proving power of type theory, while the latter is our main theorem of
`well-behavedness'.

On the semantic side, we extend the standard interpretation of type theory in
terms of sets and functions to account for erasure (\cref{def:std-model}). This
is possible when there is a tiny proposition available in the interpretation
category. We give a fully worked example in the category of families of sets
and discuss how it extends to Grothendieck toposes.

We then construct a code extraction model (\cref{def:code-extr}) producing
untyped lambda calculus terms that only retain runtime data, along with a
correctness proof by gluing that extracted code always tracks the standard
interpretation (\cref{thm:non-inter}). This is intended to be implemented as
part of the compilation pipeline post-typechecking.

We provide a demo implementation elaborating a high-level language with erased
binders similar to Idris or Agda into the core language presented here
(\cref{sec:implementation}), and extracting untyped lambda calculus terms in the
end. We include supplementary notes on the implementation of pattern unification
for metavariables. Our theoretical results are formalised in Agda; the symbol
\agdaicon{https://cthe.me/redirects/erasure/repo} appearing throughout
the paper is a hyperlink to the relevant part of the formalisation.

\section{An informal presentation of erasure} \label{sec:informal-itt}

Here we present erasure informally through a high-level dependently typed
language to be elaborated into a yet unspecified core syntax. This behaves
essentially the same as Idris/Agda with \un{0}/\un{@0} annotations. For now, we
use type-in-type for simplicity.

\subparagraph{Modes and usages}

Terms are split into the \emph{runtime} mode ($\omega$) and the
\emph{erased} mode ($0$):
\begin{itemize}
    \item runtime terms are written as $a \col{\omega} A$
    \item erased terms are written as $a \col{0} A$
    \item every runtime term can be used as an erased term
\end{itemize}
We order these as $0 < \omega$. This asymmetry reflects that at compile-time
everything is accessible but at runtime only non-erased data is. This is a
different kind of compile-time/runtime distinction from program staging or
metaprogramming. A compile-time term will not necessarily have a \emph{known
value} at compile-time; rather, it will not survive \emph{past} the compilation
phase. The runtime phase comes after the compilation phase, so any data that
reaches that point is also accessible before it -- at compile-time. Variables
are also annotated with a mode $\{0, \omega\}$. We will generally write $a : A$
for $a \col{\omega} A$, explicitly indicating only erased terms. In the core
language the subusaging rule will be implemented by an explicit coercion
mechanism (\cref{sub:phase}).

\subparagraph{Functions}\label{ssub:dep-funcs} We write $(x \col{i} A) \to B$
for the type of dependent functions from $A$ to $B$ where the domain is at mode
$i \in \{0, \omega\}$. A function with a runtime domain $(x : A) \to B$ survives
compilation, while a function with an erased domain $(x \col{0} A) \to B$ is
compiled to its return value directly. In the latter case, the rules of the
theory ensure that the return value will not have any dependency on the input
$x$ that can be observed at runtime. For example, this results in every closed
function $(n \col{0} \s{Nat}) \to \s{Nat}$ being compiled to a constant numeral
(\cref{thm:non-inter}).

\subparagraph{Universes}\label{ssub:universes} The universe of types, denoted
$\s U$, is not computationally relevant: we only need an erased code for a type
to decode it. In other words, if $A \col{0} \s U$ then $A$ is a type. For type
dependency, such as in $(x \col{i} A) \to B$, regardless of the mode $i$ of the
bound variable, it suffices for $B$ to be applied to an erased $x \col{0} A$ to
yield a type. This means there is an equivalence between type families $(x
\col{0} A) \to \s{U}$ and $(x \col{\omega} A) \to \s{U}$.

\begin{example}
 The identity function $\lambda \, \{A\} \, x.\ x : \{A \col{0} \s{U}\} \to A \to A$
  extracts to the untyped identity function $\lambda x.\ x$, erasing the type argument $A$.
  Here we freely use the notation $\{x \col{i} A\} \to B$ for an implicit
  function type, which is merely an elaboration feature.
\end{example}

\begin{example}\label{ex:vec}
The signature describing length-indexed lists, or vectors, is
\[\begin{array}{ll}
 & \s{Vec} : (A  \col{0} \s{U}) \to (n \col{0} \s{Nat}) \to \s{U} \\
 & \s{nil} : \{A  \col{0} \s{U}\} \to \s{Vec}\ A\ \s{zero} \\
 & \s{cons} : \{A \col{0} \s{U}\} \to \{n \col{0} \s{Nat}\} \to (x : A) \to \s{Vec}\ A\ n \to \s{Vec}\ A\ (\s{succ}\ n) \\
\end{array}\] This is a standard example of erasure, where vectors do not store
their length at runtime due to the usage of $0$ in the $n$ argument of
$\s{cons}$. The types are also all erased.
\end{example}

\subparagraph{Pairs}\label{ssub:dep-pairs} Similar to functions, $(x \col{i} A)
\times B$ is the type of mode-aware dependent pairs; if $i = 0$ then the
first projection is erased, and if $i = \omega$ it exists at runtime.  We can also extend this to customise the mode of
the second projection $(x \col{i} A) \times^{j} B$, which is equivalent to $(x
\col{i} A) \times ((\_ \col{j} B) \times \s{Unit})$ but with a better
runtime representation in the case where $j = 0$.

\begin{example}
Erased dependent pairs capture the situation where we have some index that needs
to be packaged existentially but which we never access at runtime. For example, we can
package the erased length of a vector, to get a list:
\[\begin{array}{l}
  \s{List}\ A: \s{U} \\
  \s{List}\ A = (n \col{0} \s{Nat}) \times \s{Vec}\ A\ n \end{array} \] At
runtime, only the second projection of an erased pair survives, so if defined
this way, lists and vectors have the same runtime representation.
\end{example}

\begin{example}
If the second projection is erased instead, this captures the situation where
we want to carry a proof about a runtime object that shouldn't exist at runtime.
We can define the type of natural numbers less than $n$ by
\[\begin{array}{l}
\s{Fin} : \s{Nat} \to \s{U} \\
\s{Fin}\ n = (k : \s{Nat}) \times_0 \s{Lt}\ k\ n \end{array} \] for an
appropriately defined less-than predicate $\s{Lt} : (k : \s{Nat}) \to (n :
\s{Nat}) \to \s{U}$.
\end{example}

\subparagraph{Inductive types}
Inductive types are carried over from MLTT; the signature in \cref{ex:vec} is a
valid such instance. The only new rule is that the mode of the scrutinee in
eliminators must be \emph{at least as strong} as the mode of the output. In
other words, we cannot pattern match on compile-time data at runtime. Besides
that, there is a new axis of choices from the ability to mark parts of inductive
types as erased. We can mark data arguments or recursive arguments as erased, or
even entire constructors. This is explored in \cref{ssub:inductive}.

\begin{example}
For booleans, given $x \col{i} \s{Bool}$ and $a,\ b \col{j} A$ where
$i \geq j$ we have
\[
  \s{if}\ x\ \s{then}\ a\ \s{else}\ b \col{j} A \,.
\]
This means that if we are in erased mode, we can perform
case analysis on any boolean, but if we are in runtime mode, we
can only perform case analysis on a runtime boolean.
\end{example}

\subsection{Erasure as a phase distinction}\label{sub:phase}

When we transition from the surface language to the core language, the
subusaging rule of erasure is explained through a \emph{synthetic phase
distinction}. The concept of `phase distinction' originates from Cardelli
\cite{Cardelli1988-zp}, later reformulated synthetically by Sterling and
Harper~\cite{Sterling2021-pm}.

A synthetic phase distinction is simply an \emph{abstract} proposition $\#$,
meaning a propositional sort of our type theory to which we do not attach a
specific truth value. We are allowed to bind proofs of $\#$ in contexts, but
there is no way to produce a closed proof of $\#$, or store proofs of $\#$
inside terms or types. If a context $\Gamma$ contains the proposition $\#$, we
write $\# \in \Gamma$, which itself is a \emph{decidable} proposition in the
syntax.

\begin{example}
Using $\rhd$ for context extension and $\bullet$
for the empty context, we have:
$$
\# \in (\bullet \rhd x \col{0} \s{Nat} \rhd \# \rhd y \col{\omega} \s{Fin}\ x) \qquad
\# \notin (\bullet \rhd x \col{0} \s{Nat} \rhd y \col{\omega} \s{Fin}\ x) 
$$
In total, we have three kinds of context extension:
extending by a runtime term variable, extending by an erased term variable, or extending
by a $\#$ proof variable.
\end{example}

We call the propositional sort $\#$ the \emph{erasure marker}, because a context
is considered erased when $\# \in \Gamma$. In such an erased context $\Gamma$,
runtime and erased terms become interchangeable through two new constructors
$\uptm$ and $\downtm$. The precise rules surrounding these are:
\begin{itemize}
    \item \label{item:marker-1} If $a \col{0} A$ and $\# \in \Gamma$,
      then ${\uptm a} \col{\omega} A$.
    \item \label{item:marker-2} If $a \col{\omega} A$ in context $\Gamma \rhd \#$, then ${\downtm a}
      \col{0} A$ in context $\Gamma$.
    \item \label{item:marker-3} $\uptm$ and $\downtm$ are mutual inverses up to definitional equality.
\end{itemize}
To provide an erased term, it suffices to provide a runtime term under the
assumption of $\#$ -- a uniqueness property similar to `every function is a
lambda' and `every product is a pair'. Crucially, $\#$ cannot be introduced or
discharged in any other way (it is not a type!). The justification for the first point above is that
$\# \in \Gamma$ holds only when we are already in a \emph{subterm} of an erased
term. This way, we emulate the \emph{phase distinction requirement} of Cardelli
that every subterm of a compile-time term should itself be compile-time.

This mechanism allows us to treat \emph{any} runtime term as erased, eliminating
the need to separately axiomatise the erased variant of a term when the runtime
variant is already present in the theory. The decidability of $\# \in \Gamma$
means we can build an elaboration algorithm which inserts all $(\uptm,\downtm)$
coercions (\cref{sec:implementation}). The high-level syntax we have presented
thus far is the input to this elaboration.

\section{Formal setup} \label{sec:formal}

In this section we formally develop type theory with erasure, assuming
familiarity with category theory and in particular categories with families
(CwFs) \cite{Castellan2019-sh}. We work in a constructive intensional type
theory with a bounded cumulative hierarchy of universes $\Set_\ell$ ($\ell \leq
\omega + 1$), natural numbers $\std N$, function extensionality, uniqueness of
identity proofs, and quotient inductive-inductive types
\cite{Altenkirch2018-nd}. We write $\re{Psh}_{\ell}(C)$ for the category of
presheaves over $C$ valued in $\Set_\ell$. In general, when we omit $\ell$, we
mean $\ell = \omega$. We also omit some equality transports for readability.
When working internally to categories, we overload type-theoretic notation; $(x
: A) \to B$ might denote a function type in a presheaf category
\cite{Hofmann1997-on} or in some theory of signatures (\cref{par:gat}). We use
$\univ$ for internal universes, mentioning this explicitly when necessary. We
write \ZTT for type theory with erasure and \TT for ordinary Martin-L\"of type
theory. Rather than fixing type formers, we develop the theory modularly. For
example, `\TT with $\Pi$-types' denotes MLTT with only function types.

\subparagraph{Generalised algebraic theories}\label{par:gat}

A generalised algebraic theory (GAT) \cite{Cartmell1986-ig} is a description of
a theory consisting of sorts, operations and equations which are allowed to
appear in arbitrary order where dependency is permitted. Such a description can
be neatly given by a context in a certain dependent type theory called the
\emph{GAT theory of signatures} \cite{Kovacs2023-gq}. This supports $\std 1$,
$\Sigma$, equality types with UIP, a universe $\univ$ for declaring sorts, and
dependent function types $\Pi$ with domain in $\univ$, and $\Pi_\s{ext}$ with an
external type (e.g. $\std N$) as domain. GATs can be used to describe type
theories in an intrinsically well-formed fashion.
\begin{example}\label{ex:cat1}
The GAT $\s{Cat}_1$ of a category with a terminal object is given by:
\[\begin{array}{l@{\ }l}
  \s{Con} &: \univ \\
  \s{Sub} &: \s{Con} \to \s{Con} \to \univ \\
  \bullet &: \s{Con} \\
  \epsilon &: \s{Sub}\ \Gamma\ \bullet \\
  \epsilon\eta &: (\sigma : \s{Sub}\ \Gamma\ \bullet) \to \sigma = \epsilon
\end{array}
\qquad
\begin{array}{l@{\ }l}
  \s{id} &: \s{Sub}\ \Gamma\ \Gamma \\
  -\circ- &: \s{Sub}\ \Phi\ \Delta \to \s{Sub}\ \Gamma\ \Phi \to \s{Sub}\ \Gamma\ \Delta \\
  \s{id}\circ &: \s{id} \circ \sigma = \sigma \\
  \circ\s{id} &: \sigma \circ \s{id} = \sigma \\
  \s{assoc} &: (\sigma \circ \tau) \circ \rho = \sigma \circ (\tau \circ \rho)
\end{array}\]
\end{example}

\begin{definition}[$\Sigma$-CwF]
A \emph{$\Sigma$-CwF}\label{def:sigma-cwf} is a CwF
equipped with $\std 1$ and $\Sigma$ types. A morphism of $\Sigma$-CwFs is a
strict CwF morphism preserving $\Unit$ and $\Sigma$, inducing a category
$\re{CwF}_{\Sigma} \subseteq \re{CwF}$.
\end{definition}

Each GAT $G$ gives rise to a \emph{freely generated} $\Sigma$-CwF, where base
types are determined by the sorts of $G$, and base terms by the operations of
$G$ quotiented by the equations in $G$. This is the approach of Bocquet
\cite{Bocquet2025-ox}, which equips GATs with \emph{functorial semantics}
\cite{Lawvere1963-qd}. We identify $G$ with its freely generated $\Sigma$-CwF.
In this language, a model\label{def:gat-model} $\cl M$ of a GAT $G$ is simply a
$\Sigma$-CwF morphism $\cl M : G \to \Set$, with the standard $\Sigma$-CwF
structure on $\Set$. To give a morphism out of $G$, it suffices to define its
actions on the signature items of $G$.

\begin{example}
A morphism $\cl F : \s{Cat}_1 \to E$ is determined by a closed type $\cl
F.\s{Con}$ in $E$ and a type $\cl F.\s{Sub}\ \Gamma\ \Delta$ over $\Gamma : \cl
F.\s{Con}$ and $\Delta : \cl F.\s{Con}$ in $E$, such that the terminal object,
identity and composition are preserved.
\end{example}

\subparagraph{Second-order generalised algebraic theories}

Second-order generalised algebraic theories (SOGATs) \cite{Uemura2021-jq,Bocquet2025-ox} are a generalisation of
GATs which allow second-order binding to appear in signatures. The \emph{SOGAT
theory of signatures} extends the GAT theory of signatures with a subuniverse
$\univ\lR \subset \univ$ and a dependent function type $\Pi\lR$ in $\univ$ with
domain in $\univ\lR$.

\begin{example}
The simplest interesting SOGAT is untyped lambda calculus with $\beta\eta$ laws:
\begin{equation}\label{eq:lc-sogat}\begin{array}{l}
  \Tm : \univR \\
  (\s{lam}, \s{app}, \beta, \eta) : (\Tm \to \Tm) \simeq \Tm \end{array}
\end{equation}
The isomorphism $(\Tm \to \Tm) \simeq \Tm$ is shorthand for two operations
$\s{lam}$, $\s{app}$ and two coherence equations $\beta$, $\eta$. The forward
direction $\s{lam}$ implicitly uses $\Pi\lR$ to bind the second-order occurrence
of $\Tm$, while the backward direction $\s{app}$ uses the first-order $\Pi$.
\end{example}

SOGATs are an even more convenient tool to describe type theories because they
already include a notion of variable binding. This avoids the boilerplate of
encoding variables and substitution manually, at least for \emph{structural}
theories (ones that do not involve linearity, contextual modalities or other
constraints on variable usage).

\begin{definition}[$(\Sigma, \Pi\lR)$-CwF]
A \emph{$(\Sigma, \Pi\lR)$-CwF}\label{def:sigma-pir-cwf} is a $\Sigma$-CwF
  with a subpresheaf $\Ty\lR \subseteq \Ty$ of \emph{representable types}
  closed under $\Unit$ and $\Sigma$, and a function type $\Pi\lR$ with domain
  in $\Ty\lR$, inducing a category $\re{CwF}_{\Sigma, \Pi\lR} \subseteq
  \re{CwF}_\Sigma$.
\end{definition} 

Similar to before, a SOGAT can be viewed as a freely generated $(\Sigma, \Pi_{\s
R})$-CwF. This time, sorts in $\univ\lR$ become the base types in $\Ty\lR$. A
model\label{def:sogat-model} $\cl M$ of a SOGAT $S$ consists of a category $\cl
M_\diamond$ with terminal object and a $(\Sigma, \Pi\lR)$-CwF morphism $\cl M :
S \to \Psh{\cl M_\diamond}$. The codomain $\Psh{\cl M_\diamond}$ has a standard
$(\Sigma, \Pi_{\s R})$-CwF structure where the representable types are presheaves
with a `context extension' operation on $\cl M_\diamond$. This includes the types of
$\cl M_{\diamond}$ if $\cl M_{\diamond}$ itself is a CwF, by the Yoneda embedding $\yoneda$. The category $\cl M_\diamond$
should be thought of as the underlying \emph{category of contexts} of $\cl M$, where
second-order binding in $S$ is interpreted using context extension along
representable types in $\cl M_\diamond$. Once again, to give a morphism out of $S$
it suffices to define its actions on the signature items of $S$.

\subparagraph{Correspondence between GAT and SOGAT models}

Bocquet's approach to the model theory of SOGATs is to reduce them to GATs and
then reuse the model theory of GATs. From a SOGAT $S$ we can compute a GAT
$S^{\s{fo}}$ extending $\s{Cat}_1$, such that SOGAT models $\cl M : S \to \Psh
{\cl M_{\diamond}}$ are in bijective correspondence with GAT models
$\widetilde{\cl M} : S^\s{fo} \to \Set$ where the $\s{Cat}_1$ part of $S^\s{fo}$
maps to $\cl M_{\diamond}$. We often reuse the name $\cl M$ for $\widetilde{\cl
M}$ as it is unambiguous to do so. The mapping $S \mapsto S^{\s{fo}}$
essentially corresponds to the signature-based translation from SOGATs to GATs
detailed by Kaposi and Xie~\cite{Kaposi2024-db}: it maps a SOGAT $S$ to the GAT
of a category with terminal object and the presheaf interpretation of $S$ over
it. This generates all the variable and substitution boilerplate machinery
necessary to interpret the second-order binding in $S$ in terms of first-order
context extension. This mapping is \emph{functorial}: a SOGAT morphism $S \to T$
yields a GAT morphism $S^{\s{fo}} \to T^{\s{fo}}$ that preserves the category
structure.

\begin{example}\label{ex:lc-fo} Translating the SOGAT of untyped lambda calculus
from \cref{eq:lc-sogat} produces a GAT extending $\s{Cat}_1$ with a family $\Tm
: \Con \to \univ$ of terms with substitution $-[-] : \Tm\ \Delta \to \Sub\
\Gamma\ \Delta \to \Tm\ \Gamma$ preserving identity and composition (in other
words, a presheaf), a context extension operator $-\rhd : \Con \to \Con$
representing the presheaf $\Tm$ via $\s{Sub}\ \Gamma\ (\Delta \rhd) \simeq
\s{Sub}\ \Gamma\ \Delta \times \Tm\ \Gamma$, and an isomorphism $(\s{lam},
\s{app}) : \Tm\ (\Gamma \rhd) \simeq \Tm\ \Gamma$ encoding the second-order
constructor $\s{lam}$ using context extension.
\end{example}

\subparagraph{Categories of models}

For any two models $\cl M, \cl N$ of a GAT $G$, there is a function space model
$\re{Func}(\cl M, \cl N)$ of natural transformations from $\cl M$ to $\cl N$,
displayed over $G$. For example, a context displayed over $\Gamma$ in
$\s{Func}(\cl M, \cl N)$ is a function $\cl M\, \Gamma \to \cl N\, \Gamma$. A
strict morphism of models $\cl M \to \cl N$ is thus defined as a dependent
$\Sigma$-CwF morphism $G \to \re{Func}(\cl M, \cl N)$. This yields the usual
notion of GAT morphism that can be computed from signatures
\cite{Kovacs2023-gq}, and induces a category of $G$-models $\re{Mod}(G)$. The
syntax $\re 0_G$ of a GAT $G$ is the initial $G$-model, which always exists and
is given by a quotient inductive-inductive type; for any other $G$-model $\cl
M$, we write $\llbracket-\rrbracket_{\cl M} : \re 0_G \to \cl M$ for the unique
$G$-model morphism. The syntax $\re 0_S$ of a SOGAT $S$ is simply $\re 0_{S^{\s
fo}}$. A strict morphism\label{rem:morphisms} of GATs $\cl F : G \to H$ yields a
functor between categories of models $\cl F^* : \re{Mod}(H) \to \re{Mod}(G)$ by
precomposition. Applying this to $\re 0_H$, we get a
$G$-model $\cl F^*\ \re 0_H$. We often abuse notation and write $\cl F$ instead of $\llbracket - \rrbracket_{\cl F^* \re 0_H}$.

\section{\ZTT as a SOGAT}\label{sec:itt-sogat}

Now we proceed to define \ZTT as a SOGAT and compute its GAT translation. The
full definition is given in \cref{fig:tt0-sogat}, with the extension to
inductive types described in \cref{ssub:inductive}.

\begin{figure}[h]
\centering
\[\small\begin{array}{l}
  \s{Mode} \triangleq \{0, \omega\} \\
  \Ty : \std{N} \to \univ \\
  \Tm : \s{Mode} \to \Ty_\ell \to \univR \\
  \\
  \# : \univR \\
  \#\s{-prop} : (p, q : \#) \to p = q \\
  (\downtm, \uptm) : (\# \to \Tm_\omega\ A) \simeq \Tm_0\ A
\end{array}
\quad
\begin{array}{l}
  \Pi_i : (A : \Ty_\ell) \to (\Tm_i\ A \to \Ty_\ell) \to \Ty_\ell \\
  (\s{lam}, \s{app}) : ((x : \Tm_i\ A) \to \Tm_\omega\ (B\ x)) \simeq \Tm_\omega\ (\Pi_i\ A\ B) \\
  \\
  \Sigma_i : (A : \Ty_\ell) \to (\Tm_i\ A \to \Ty_\ell) \to \Ty_\ell \\
  (\s{pair}, \s{proj}) : ((x : \Tm_i\ A) \times \Tm_\omega\ (B\ x)) \simeq \Tm_\omega\ (\Sigma_i\ A\ B) \\
  \\
  \s{U} : (\ell : \std{N}) \to \Ty_{\ell + 1} \\
  (\s{El}, \s{code}) : \Tm_0\ \s{U}_\ell \simeq \Ty_\ell
\end{array}\]
\caption{The SOGAT defining \ZTT with $\Pi$, $\Sigma$, and universes (\agdaicon{https://cthe.me/redirects/erasure/ttwE}).}
\label{fig:tt0-sogat}
\end{figure}

\ZTT is a theory involving three sorts: types $\Ty$ externally indexed by
universe levels in $\std N$, terms $\Tm$ externally indexed by modes $\{0, \omega\}$ and
internally indexed by types, and the \emph{erasure marker} $\#$, forced to be a
proposition by $\#\s{-prop}$. The definitional isomorphism $(\downtm, \uptm)$
converts between modes: if $a \col{0} A$ and $p : \#$, then $\uptm[p] a
\col{\omega} A$; if $a \col{\omega} A$ under assumption $p: \#$, then
$\downtm[p.] a \col{0} A$. In $\uptm[p] a$, $p$ is an argument, while in
$\downtm[p.] a$, $p$ is \emph{bound} in $a$, shorthand for $\downtm (\lambda p.\
a)$. Since we can always weaken $a : \Tm_\omega\ A$ to $\lambda\_.a : \# \to
\Tm_\omega\ A$, we can convert $\omega$ terms to 0 terms in any context. By
induction on $i$, we implicitly extend this to $\downtm : \Tm_i\ A \to \Tm_0\
A$.

We include $\Sigma$ and $\Pi$ types whose introduction and elimination rules
involve terms in mode $\omega$, while type formation binds a term of mode $i$
matching the domain mode. Universes have codes \emph{in mode 0}. This induces an
isomorphism $(\# \to \s{Ty}_\ell) \simeq \s{Ty}_\ell$ which allows us to convert
between type families $\Tm_\omega\ A \to \Ty_\ell$ and $\Tm_0\ A \to \Ty_\ell$.
Without universes we would have to include this isomorphism as a primitive. This
yields an alternative equivalent formulation of $\Pi_i$ and $\Sigma_i$ 
where the binder in the type formation rule is always indexed by $\Tm_0\ A$, for
example: $\Sigma_i : (A : \Ty_\ell) \to (\Tm_0\ A \to \Ty_\ell) \to \Ty_\ell$.
Also, it is possible and practically desirable to include strict cumulativity
\cite{Sterling2019-ze} for universes, which we omit here.

\subparagraph{The erased fragment}\label{ssub:erased}

The theory only explicitly includes terms for the runtime fragment
$\omega$. Usually, presentations of theories with erasure will include
a copy of terms in each mode. In our case, this is not necessary. By virtue
of the isomorphism $(\downtm, \uptm)$, we automatically get the entire erased
fragment `for free'.
For example, erased terms for function types at any domain mode are derivable:
\[
\begin{array}{l}
  \s{lam}_0 : ((x : \Tm_0\ A) \to \Tm_0\ (B\ x)) \to \Tm_0\ (\Pi_i\ A\ (B \circ \downtm)) \\
  \s{lam}_0\ f = {\downtm[p.] \s{lam}\ (\lambda x.\ \uptm[p] (f\ \downtm x))} \\[1em]
  \s{app}_0 : \Tm_0\ (\Pi_i\ A\ B) \to (x : \Tm_0\ A) \to \Tm_0\ ((B \circ \uptm)\ x)  \\
  \s{app}_0\ t\ x = {\downtm[p.] \s{app}\ (\uptm[p] t)\ (\uptm[p] x)}
\end{array}\]
Here we use $- \circ \uptm$ and $- \circ \downtm$ for the two directions
of $(\Tm_\omega\ A \to \Ty_\ell) \simeq (\Tm_0\ A \to \Ty_\ell)$.
The isomorphism providing $\beta$ and $\eta$ for $(\s{lam},
\s{app})$ also holds for $(\s{lam}_0, \s{app}_0)$. The idea is that if the goal
is of the form $\Tm_0\ X$, then we introduce $\downtm[p.]$ and reduce the goal
to $\Tm_\omega\ X$ while having access to $p$. Then we use whatever runtime term
former to fulfil the goal. When such a runtime term former \emph{expects} a
$\Tm_\omega\ Y$  but all we have is a $\Tm_0\ Y$, we wrap it in $\uptm[p]$ using
the $p$ we bound earlier. We can derive the rest of the erased fragment this
way. From now on, we use the subscript $_0$ as in $\s{lam}_0$ to denote the
erased variant of some term former.

\subparagraph{Inductive types}\label{ssub:inductive}

It is straightforward to include any indexed inductive type in \ZTT. We only
need to include it in mode $\omega$, and then we can derive its erased fragment
using the technique in \cref{ssub:erased}. However, the main affordance of
erasure is that we can choose for some constructor data, usually indices, to
\emph{always} be erased. For example, here is the inductive family of
length-indexed vectors, where the length is considered erased data:
\[
\begin{array}{l}
  \s{Vect} : \s{Tm}_0\ {\s{Nat}} \to \Ty_\ell \to \Ty_\ell \\
  \s{nil} : \Tm_\omega\ (\s{Vect}\ \s{zero}\ A) \\
  \s{cons} : \{k : \Tm_0\ \s{Nat}\} \to \Tm_\omega\ A \to
    \Tm_\omega\ (\s{Vect}\ k\ A) \to \Tm_\omega\ (\s{Vect}\ (\s{succ}_0\ k)\ A) \\
  \s{elim} : (P : (n : \Tm_0\ \s{Nat}) \to \Tm_0\ (\s{Vect}\ n\ A) \to \Ty_\ell) \\
  \quad \to \Tm_\omega\ (P\ \s{zero}_0\ \s{nil}_0) \\
  \quad \to (\{k : \Tm_0\ \s{Nat}\} \to (x : \Tm_\omega\ A)
  \to (\textit{xs} : \Tm_\omega\ (\s{Vect}\ k\ A))  \\
  \quad \quad \to \Tm_\omega\ (P\ k\ \textit{xs})
  \to \Tm_\omega\ (P\ (\s{succ}_0\ k)\ (\s{cons}_0\ x\ \textit{xs}))) \\
  \quad
  \to
  \{n : \Tm_0\ \s{Nat}\} \to
  (v : \Tm_\omega\ (\s{Vect}\ n\ A))
  \to \Tm_\omega\ (P\ n\ v) \\
\end{array}\]
We have omitted type indexing and computation rules for brevity. Erased constructor
arguments that appear in return indices can sometimes be converted to
\emph{relevant} arguments. A sufficient condition for this is Brady et~al.'s
forcing analysis \cite{Brady2004-ay}, which can be applied here: if we make the
$n$ in the output above \emph{relevant}, we can also make $k$ in the $\s{cons}$
method relevant because $k$ can be computed at runtime by stripping $\s{succ}$
from $n$. It is possible to extend code extraction (\cref{def:code-extr}) with such
`forced' eliminators, and other optimisations in \cite{Brady2004-ay} like
detagging.

We could include general W-types \cite{Hugunin2021-ui} in \ZTT, but this would
not afford us the most generality. This is because we have a choice of erased
data both in the non-recursive and the recursive arguments. The former would be
possible with standard W-types because we have mode-dependent $\Sigma$ types,
but the latter wouldn't. For example, consider the type
\[
\begin{array}{l}
  \s{Nat0} : \Ty_\ell \\
  \s{zero0} : \Tm_\omega\ \s{Nat0} \\
  \s{succ0} : \Tm_0\ \s{Nat0} \to \Tm_\omega\ \s{Nat0} \\
\end{array}\]
This cannot be encoded as a regular W-type because W-types encode \emph{arities}
of recursive arguments, but now there is also the axis of $0/\omega$ to choose
from. Additionally, we might want entire constructors to only exist in mode 0:
\[
\begin{array}{l}
  \s{Bool}_{\omega\to \s{true}} : \Ty_\ell \\
  \s{true} : \Tm_\omega\ \s{Bool}_{\omega\to \s{true}} \\
  \s{false} : \Tm_0\ \s{Bool}_{\omega\to \s{true}} \\
\end{array}\]
This is the type of booleans that are always true at runtime. So there is also
an axis of $0/\omega$ for the return sort of each constructor. Altogether these
`exotic' inductive types suggest that it would be useful to investigate
\emph{signatures} for inductive types with erasure to better understand this
space of possibilities. We leave this to future work.

\subparagraph{Compatibility with two-level type theory}\label{ssub:2ltt}

Having formulated erasure as a SOGAT, we can now easily combine it with other
features formulated as SOGATs. For example, we can form a theory $\s{2LTT}_0$ of
two-level type theory \cite{Kovacs2022-rf} with erasure. Aside from its use in
HoTT \cite{Annenkov2023-vk}, two-level type theory can be viewed as a type
theoretic formulation of staging/metaprogramming. It allows us to control which
binders are evaluated at compile-time versus runtime, existing in the same space
of features as erasure -- features that aim to improve the feasibility of
dependent types for practical programming.

The general principle of two-level type theory is that some or all judgments in
the object theory become type formers in the meta-theory. In this case, our
object theory is type theory with erasure. Therefore, we promote $\Tm_\omega$
and $\Tm_0$ from sorts to meta-level types (but not $\#$, to retain the
decidability of its presence in any context):
\[\begin{array}{l@{\qquad}l}
\s{Ty}^{\s M} : \univ & {\Uparrow} : (i : \s{Mode}) \to \s{Ty} \to \s{Ty}^{\s M} \\
\s{Tm}^{\s M} : \s{Ty}^{\s M} \to \univR & (\langle-\rangle, {\sim}-) :
\s{Tm}^{\s M}\ {\Uparrow}_i A \simeq \s{Tm}_i\ A \end{array}\]
Here the
superscript $\s M$ denotes the meta fragment, while the base theory corresponds
to the object fragment. In the meta fragment, we have two separate lifting
operations $\Uparrow_i$, one for each mode $i$. The erasure marker $\#$ and
coercions $\downtm/\uptm$ remain unchanged from before. From them we can derive
coercions in the meta level as well:
\[\begin{array}{l}
  ({\downtm}^{\s M}, {\uptm}^{\s M}) : (\# \to \s{Tm}^{\s M}\
  {\Uparrow}_\omega A) \simeq \s{Tm}^{\s M}\ {\Uparrow}_0 A
\end{array}\]
These are given by $\downtm^{\s M}\ f = \langle {\downtm}_{p.}\ {\sim}(f\ p)
\rangle$ and $\uptm^{\s M}_p\ x = \langle \uptm_p\ {\sim} x \rangle$.

\subsection{Generated first-order theory (\agdaicon{https://cthe.me/redirects/erasure/cwfwE})}\label{sub:first-order}

Using the results from \cref{sec:formal}, we now compute the GAT specification
of \ZTT. This is the actual `type system' in the traditional sense, which
includes contexts and variables. As with any SOGAT, we start with a category
with terminal object described in \cref{ex:cat1}.

\subparagraph{Sorts}

For each sort in the SOGAT, we get a presheaf on this category:
\[
\Ty : \std{N} \to \Con \to \univ \qquad \Tm : \s{Mode} \to (\Gamma : \Con) \to
\Ty_\ell\ \Gamma \to \univ \qquad {\#\in-} : \Con \to \univ
\]
There is also an operation $(x, y : \# \in \Gamma) \to x = y$ from $\#\s{-prop}$
that forces $\# \in -$ to be proposition-valued. We now see that $\#\in\Gamma$
(\cref{sub:phase}) is simply the presheaf generated by the sort $\#$ evaluated
at a context $\Gamma$. Each of these presheaves comes with a substitution
operation $A[\sigma]$, $t[\sigma]$, $\pi[\sigma]$ which respects $\circ$ and
$\s{id}$.

\subparagraph{Context extensions}
The representability of mode-indexed terms and the erasure marker yields
the context extension operations
\[-\rhd_{i}- : (\Gamma : \Con) \to \Ty_\ell\ \Gamma \to \Con \qquad
-\rhd \# : \Con \to \Con\]
with corresponding isomorphisms
\[\begin{array}{ll}
((-,_i-),\, \s{pq}_i) &: (\sigma : \Sub\ \Gamma\ \Delta) \times \Tm_i\ \Gamma\ A[\sigma] \simeq \Sub\ \Gamma\ (\Delta \rhd_i A) \\
((-,_\#-),\, \s{pq}_\#) &: \Sub\ \Gamma\ \Delta \times (\# \in \Gamma) \simeq
\Sub\ \Gamma\ (\Delta \rhd \#) \end{array}\] defining that substitutions can be
extended by terms and erasure marker witnesses, and that terms in both modes as well
as the erasure marker have the `0th' de~Bruijn index by $\s{q}$ and weakening
by $\s{p}$. More explicitly, we get
\[\begin{array}{ll@{\qquad}ll}
\s{p}_i &: \Sub\ (\Gamma \rhd_i A)\ \Gamma &
\s{p}_\# &: \Sub\ (\Gamma \rhd \#)\ \Gamma \\
\s{q}_i &: \Tm_i\ (\Gamma \rhd_i A)\ A [\s{p}_i] &
\s{q}_\# &: \# \in (\Gamma \rhd \#)
\end{array}
\]
which are derivable from the isomorphisms above, satisfying the usual CwF rules.

\subparagraph{Erasure coercions}

The defining isomorphism of $\#$ is presented in the GAT as
\[(\downtm, \uptm) : \Tm_\omega\ (\Gamma \rhd \#)\ A \simeq \Tm_0\ \Gamma\ A\]
relating runtime terms in a context extended by $\#$ to erased terms. The
$\uptm$ direction outputs a term in an extended context. We can also derive
a form which stores the data of the substitution needed to make the context general:
\[{\uptm'} : \# \in \Gamma \to \Tm_0\ \Gamma\ A \to \Tm_\omega\ \Gamma\ A \qquad
{\uptm[p]'} t = (\uptm t)[\s{id} \, ,_\# \, \s{p}_\#]\]
Just like before, we can extend $\downtm$ to operate on a term in any mode $i$ -- now explicitly:
\[{\downtm[*]} : \Tm_i\ \Gamma\ A \to \Tm_0\ \Gamma\ A \qquad
{\downtm[*]}\ \{i=0\}\ t = t \qquad
{\downtm[*]}\ \{i=\omega\}\ t = {\downtm} (t[\s{p}_\#])\]

\subparagraph{Standard types}

The rest of the theory is translated to a first-order representation that looks
very similar to the usual CwF structures. For $\Pi$ types we get
\[\begin{array}{ll}
  \Pi_i : (A : \Ty_\ell\ \Gamma) \to (B : \Ty_\ell\ (\Gamma \rhd_i A)) \to \Ty_\ell\ \Gamma \\
  (\s{lam},\s{app}) : \Tm_\omega\ (\Gamma \rhd_i A)\ B
\simeq \Tm_\omega\ \Gamma\ (\Pi_i\ A\ B) \end{array}\]
matching the usual CwF formulation: the codomain $B$ lives over the
mode-$i$ extension $\Gamma \rhd_i A$. We can write the less `categorical'
application operator as
\[\s{app}' : \Tm_\omega\ \Gamma\ (\Pi_i\ A\ B) \to (x : \Tm_i\ \Gamma\ A) \to
\Tm_\omega\ \Gamma\ B[\langle x \rangle]\] where $\langle-\rangle
\triangleq (\s{id},_i -) : \Tm_i\ \Gamma\ A \to \Sub\ \Gamma\ (\Gamma \rhd_i A)$
is the single term substitution. The derivable erased fragment
(\cref{ssub:erased}) is also carried over to the first-order GAT presentation.

\section{Syntactic properties}\label{sec:models-of-itt}

In this section, we explore some properties of the first-order syntax of $\ZTT$.

\subparagraph{Zeroing}
The first property we show is that erased terms do not depend on runtime
variables. In particular, an erased term in $\Gamma$ should uniquely correspond
to an erased term in $0\Gamma$, which is $\Gamma$ with the modes of all the
variables set to 0, and without any erasure markers $\#$. This verifies that the
phase distinction truly prevents erased data from depending on runtime data. We
do this by showing that the erased fragment of $\ZTT$ is an appropriate model
for the entire theory (the \emph{zeroing} model), such that zeroing erased $\ZTT$
terms is a bijection.

\begin{definition}[Zeroing \agdaicon{https://cthe.me/redirects/erasure/zeroing}]\label{def:zeroing-model}
  The \emph{zeroing}
  $(\Sigma, \Pi_{\s R})$-CwF endomorphism $0 : \ZTT \to \ZTT$ uses the erased fragment of
  \ZTT to implement the relevant fragment. It sets the mode of all terms to
  erased, and removes any erasure markers $\#$. It is defined by its action on the signature components:
  \[\begin{array}{ll@{\qquad\qquad}ll}
  0.\Ty_\ell      &\triangleq \Ty_\ell  &
  0.\Tm_0\ A      &\triangleq \Tm_0\ A  \\
  0.\#            &\triangleq \std{1}  &
  0.\Tm_\omega\ A &\triangleq \Tm_0\ A
  \end{array}\]
  The rest of the signature is implemented by the erased fragment (\cref{ssub:erased}).
  For example $0.\Pi_i\ A\ B \triangleq \Pi_i\ A\ B$, $0.\s{lam}\ f
  \triangleq \s{lam}_0\ f$ and $0.\s{app}\ f\ x \triangleq \s{app}_0\ f\ x$. The
  $\uptm$/$\downtm$ become trivial.
\end{definition}

By acting on the syntax, this yields a morphism of $\ZTT^{\s{fo}}$ models $\llbracket-\rrbracket_0 : \re 0_{\ZTT}
\to 0^*\ \re 0_{\ZTT}$. We can compute the actions on syntactic sorts as
$0_{\Con} : \Con \to \Con$,
$0_{\Ty} : \Ty\ \Gamma \to \Ty\ 0\Gamma$,
$0_{\Tm_\omega} : \Tm_\omega\ \Gamma\ A \to \Tm_0\ 0\Gamma\ 0A$,
$0_{\Tm_0} : \Tm_0\ \Gamma\ A \to \Tm_0\ 0\Gamma\ 0A$,
and
$0_{\#} : \# \in \Gamma \to \std{1}$.

\begin{theorem}[Types and erased terms need nothing \agdaicon{https://cthe.me/redirects/erasure/need-nothing}]\label{thm:need-nothing}
  In the syntax of \ZTT there
  exists a substitution $\upsub : \Sub\ \Gamma\ 0\Gamma$ which becomes the identity
  under zeroing, such that $(0A)[\upsub] = A$ for types and $(0a)[\upsub] = \downtm[\star] a$ for terms.
  This induces natural isomorphisms:
  \[\begin{array}{l}
    \Ty_\ell\ 0\Gamma \simeq \Ty_\ell\ \Gamma \qquad\qquad
    \Tm_0\ 0\Gamma\ 0 A \simeq \Tm_0\ \Gamma\ A
  \end{array}\]
\end{theorem}

\subparagraph{Conservativity over type theory}
We can also characterise the relationship between \ZTT and \TT. In particular,
we would like to ensure that erasure does not allow the production of any exotic
terms that are not possible in ordinary type theory. The formal notion of this
condition is \emph{conservativity}: if there exists a proof of a $\TT$-theorem
in $\ZTT$, then a proof also exists in $\TT$. This can be shown by constructing
bidirectional interpretations.

\begin{definition}[\ZTT to \TT{} \agdaicon{https://cthe.me/redirects/erasure/tt-model}]\label{def:tt-model}
  The morphism $\lemb - :
  \ZTT \to \TT$ is constructed by using \TT types and terms to implement \ZTT
  types and terms (both runtime and erased). In this model, mode annotations and
  the erasure marker are forgotten:
  \[\begin{array}{ll@{\qquad\qquad}ll}
  \lemb-.\Ty_\ell      &\triangleq \Ty_\ell  &
  \lemb-.\Tm_0\ A      &\triangleq \Tm\ A  \\
  \lemb-.\#            &\triangleq \std{1}  &
  \lemb-.\Tm_\omega\ A &\triangleq \Tm\ A
  \end{array}\]
\end{definition}

\begin{definition}[\TT to \ZTT{} \agdaicon{https://cthe.me/redirects/erasure/0tt-model}]\label{def:0tt-model}
  Conversely, the morphism $\emb - : \TT \to \ZTT$ is constructed
  by using the erased fragment to implement the structure of \TT:
  \[
  \emb-.\Ty_\ell \triangleq  \Ty_\ell \qquad\qquad
  \emb-.\Tm\ A \triangleq  \Tm_0\ A
  \]
\end{definition}

The composition $\emb{\lemb-}$ is almost equivalent to zeroing $0$ (but doesn't
preserve $\Pi$ and $\Sigma$ modes), while $\lemb{\emb-}$ is the identity. In the
terminology of Bocquet \cite[p.~40]{Bocquet2025-ox}, $\lemb -$ is a \emph{trivial
fibration}, yielding the following properties:

\begin{theorem}[Erased conservativity of \ZTT over \TT{}  \agdaicon{https://cthe.me/redirects/erasure/erased-cons}]\label{thm:erased-cons}
  \ZTT's erased fragment is conservative over \TT: there is a surjective natural map
  $\re 0_{\ZTT}.\Tm_0\ \emb \Gamma\ \emb A \to \re 0_\TT.\Tm\ \Gamma\ A$.
\end{theorem}

\begin{corollary}[Runtime conservativity of \ZTT over \TT{}  \agdaicon{https://cthe.me/redirects/erasure/rt-cons}]\label{thm:rt-cons}
  \ZTT's runtime fragment is conservative over \TT: given a \ZTT context $\Gamma'$
  and type $A'$ such that $0\Gamma' = \emb \Gamma$ and $0A' = \emb A$, there is a natural map
  $\re 0_{\ZTT}.\Tm_\omega\ \Gamma'\ A' \to \re 0_\TT.\Tm\ \Gamma\ A$.
\end{corollary}

We might be tempted to ask for a stronger conservativity result for erased
terms, namely that this map is actually an isomorphism. This is not possible if
we have mode-aware binder types. For example, consider for each mode $i$ the
term $\s{pair}\ (\s{code}\ (\Pi_i\ \s{Nat}\ \s{Nat}))\ (\s{lam}_0\ \s{q}_0) :
\Tm_0\ {\bullet}\ (\Sigma_0\ \s{U}\ (\s{El}\ \s{q}_0))$. For either choice of
$i$ we get the same \TT term, so $\lemb-$ is not injective even on erased terms.

\section{Standard models} \label{ssub:std-model}

The standard semantics of Martin-L\"of type theory are in $\Set$, and can be
generalised to presheaf categories \cite{Hofmann1997-on} or Grothendieck toposes
\cite{Gratzer2022-gn}. Now we explore standard models for $\ZTT$, first at the
level of sets, and then briefly at the generality of Grothendieck toposes.
The utility of such models, besides showing that the theory can be modelled by known
and common mathematical objects, lies in their use when proving metatheorems
via gluing \cite{Kaposi2019-lp}.

The involvement of a phase distinction in \ZTT requires the standard semantics
to take place not just in plain sets, but rather in some kind of
\emph{phase-separated} sets. The isomorphism $(\# \to \Tm_\omega\ A) \simeq
\Tm_0\ A$ shows us that in the semantics we must have an object $\#$ such that
exponentiating by $\#$ \emph{progresses} the phase from runtime to erased. It is
already known that languages with phase separations have semantics in
glued categories \cite{Sterling2022-ym}. A glued category is a comma category of
the form $\s{Id}_C \downarrow F$ for a functor $F : D \to C$. In the simplest
case, we take $C = D = \Set$ and $F = \s{Id}_{\Set}$. This yields the arrow
category of $\Set$, which is the category of presheaves over the interval $0 \to
1$, and also equivalent to the category of families of sets $\s{Fam}(\Set)$.
We choose to work with $\s{Fam}(\Set)$, using its internal language when convenient.

The category $\s{Fam}(\Set)$ serves as the simplest standard model of \ZTT,
which we denote by $\cl S$. For a family $(X_0 : \Set) \times (X_1 : X_0
\to \Set)$, the base $X_0$ stores the erased data, and the fibers store the
runtime data. Consider the object $\phi \triangleq (\std{1}, \lambda\_.\
\std{0})$, the family with a single empty fiber. It is a proposition in the
sense that any two maps into $\phi$ are equal. Exponentiation of $X = (X_0,
X_1)$ by $\phi$ acts as $(\phi \to X) \simeq (X_0,\ \lambda\_.\ \std{1})$, so
maps out of $\phi$ isolate the base component, trivialising the fibers. This
suggests that the erasure marker $\#$ should be interpreted as $\phi$. Because
we are now in a semantic setting, we can interpret erased terms directly as
maps from $\phi$ to runtime terms: $\cl S.\Tm_0\ A = (\phi \to \cl S.\Tm_\omega\ A)$.

\begin{definition}[$P$-modal]\label{def:phi-modal}
  Given a proposition $P$, an object $X$ is \emph{$P$-modal} if the
  weakening map $\s p : X \to (P \to X)$ defined by $\s p\ x\ \_ = x$
  is an isomorphism.
\end{definition}

In $\s{Fam}(\Set)$, an object is $\phi$-modal if its fibers are contractible. A
consequence of this approach is that any part of the language that is erased
should be $\phi$-modal. This notably includes universes. The standard
Hofmann-Streicher \cite{Hofmann1997-pt} universe construction can be performed
in $\s{Fam}(\Set)$, yielding the universe $\univ_\ell \triangleq (\Set_\ell,
\lambda A.\ A \to \Set_\ell)$ for $\ell < \omega$. However, this is not
$\phi$-modal, since its fibers are not contractible. We need an alternative
universe construction. Let us suggestively denote this by
$\amazing{\phi}{\univ_\ell}$ at level $\ell$. We will revisit this notation in
\cref{ssub:mod-gro}. For the $\s{El}/\s{code}$ isomorphism of \ZTT, we have:
\[
  \cl S.\s U_\ell \triangleq \amazing{\phi}{\univ_\ell} \qquad
  \cl S.\s{code}\ (A :  \amazing{\phi}{\univ_\ell}) \triangleq \s{p}\ A : \phi \to  \amazing{\phi}{\univ_\ell}  \qquad
  \cl S.\s{El}\ (c : \phi \to  \amazing{\phi}{\univ_\ell}) \triangleq \boxed{?} : \amazing{\phi}{\univ_\ell} 
\]
An assignment to $\boxed{?}$ such that $\s{El}$ and $\s{code}$ form an isomorphism
is possible when $\amazing{\phi}{\univ_\ell}$ is $\phi$-modal.
Luckily, $\s{Fam}(\Set)$ supports a universe which is $\phi$-modal:
\[
  \amazing{\phi}{\univ_\ell} \triangleq ((A : \Set_\ell) \times (A \to \Set_\ell), \lambda\_.\std 1)
\]
The decoding map of this universe is of the form $[-] :
\amazing{\phi}{\univ_\ell} \to \univ_\ell$, and is defined by first projection
for the base, and second projection for the fibers. It supports all base types
of $\univ_\ell$, and is closed under dependent products and sums. We call this
universe \emph{squashed} because it is $\phi$-modal but still retains all the
original structure. Now we are ready to define the full model:

\begin{definition}[$\s{Fam}(\Set)$ model of \ZTT{} \agdaicon{https://cthe.me/redirects/erasure/famset-so}]\label{def:std-model} The standard
 model of \ZTT in families of sets is $\cl S : \ZTT \to \re{Psh}_{\omega + 1}(\s{Fam}(\Set))$,
 given by:
  \[
  \begin{array}{ll@{\qquad\qquad}ll}
    \cl S.\Ty_\ell &\triangleq \amazing{\phi}{\univ_\ell} &
    \cl S.\Tm_0\ A &\triangleq  \phi \to [A] \\
    \cl S.\# &\triangleq \phi &
    \cl S.\Tm_\omega\ A &\triangleq [A]
  \end{array}
  \]
For erased functions we interpret $\cl S.\Pi_0\ A\ B \triangleq \re (a : (\phi \to A))
\to B\ a$ and for runtime functions we interpret $\cl S.\Pi_\omega\ A\ B \triangleq
\re (a : A) \to B\ (\s p\ a)$, both using the function type in
$\amazing{\phi}{\univ_\ell}$. Because $B : (\phi \to A) \to
\amazing{\phi}{\univ_\ell}$ in both cases, we must use the weakening map $\s p$
in the $\omega$ case to `forget' the runtime data of $a$.
\end{definition}

\subparagraph{The $\s{Fam}(\Set)$ model, explicitly (\agdaicon{https://cthe.me/redirects/erasure/famset-fo})}\label{ssub:famset}

We expand $\cl S$ in first-order morphism form: each context $\Gamma$ is interpreted as a set
$\llb\Gamma^0_{\cl S} : \Set$ (erased phase) and a family $\llb\Gamma^1_{\cl S} : \llb\Gamma^0_{\cl S} \to \Set$
(runtime phase). A substitution $\Gamma \to \Delta$ is a function $\sigma_0 :
\llb\Gamma^0_{\cl S} \to \llb\Delta^0_{\cl S}$ (erased) and a family of functions $\sigma_1 : \forall \gamma_0.\
\llb\Gamma^1_{\cl S}\ \gamma_0 \to \llb\Delta^1_{\cl S}\ (\sigma_0\ \gamma_0)$ (runtime), displayed over
the erased function. An $\ell$-type in context $\Gamma$ is interpreted as a
family of $\ell$-sets $\llb\Gamma^0_{\cl S} \to (A_0 : \Set_\ell) \times (A_1 : A_0 \to
\Set_\ell)$ indexed only over $\llb\Gamma^0_{\cl S}$. This is the result of expanding a map
into the squashed universe. An erased term of type $A$ in context $\Gamma$ is
interpreted as a section of the erased components $(\gamma_0 : \llb\Gamma^0_{\cl S}) \to
\llb A^0_{\cl S}\ \gamma_0$, while a runtime term is interpreted as a full section $(\sigma_0
: (\gamma_0 : \llb\Gamma^0_{\cl S}) \to \llb A^0_{\cl S}\ \gamma_0) \times (\sigma_1 : \forall \gamma_0.\
\llb\Gamma^1_{\cl S}\ \gamma_0 \to \llb A^1_{\cl S}\ (\sigma_0\ \gamma_0))$. The sort $\# \in \Gamma$ is
interpreted as $\llb\Gamma_{\cl S} \to \phi$ which is equivalent to $\forall
\gamma_0.\ \llb\Gamma^1_{\cl S}\ \gamma_0 \to \std{0}$; the assertion that the runtime part
of the context is uninhabited. Finally, we can compute the context formers as:
\[\begin{array}{l@{\ }l@{\ }l}
\llb{\bullet}_{\cl S} &\triangleq \std{1}
  &\simeq (\std{1}, \lambda\_.\ \std{1}) \\
\llb{\Gamma \rhd_0 A}_{\cl S} &\triangleq (\gamma : \llb\Gamma_{\cl S}) \times (\phi \to \llb A_{\cl S}\ (\s{p}\ \gamma))
  &\simeq ((\gamma_0 : \llb\Gamma^0_{\cl S}) \times \llb A^0_{\cl S}\ \gamma_0,\ \lambda (\gamma_0, a_0).\ \llb\Gamma^1_{\cl S}\ \gamma_0) \\
\llb{\Gamma \rhd_\omega A}_{\cl S} &\triangleq (\gamma : \llb\Gamma_{\cl S}) \times \llb A_{\cl S}\ (\s{p}\ \gamma)
  &\simeq ((\gamma_0 : \llb\Gamma^0_{\cl S}) \times \llb A^0_{\cl S}\ \gamma_0,\\
  &&\quad \lambda (\gamma_0, a_0).\ (\gamma_1 : \llb\Gamma^1_{\cl S}\ \gamma_0) \times \llb A^1_{\cl S}\ \gamma_0\ a_0) \\
\llb{\Gamma \rhd \#}_{\cl S} &\triangleq \llb\Gamma_{\cl S} \times \phi &\simeq (\llb\Gamma^0_{\cl S}, \lambda\_.\ \std{0})
\end{array}\]
Erased context extension only extends the erased part of the context, runtime
context extension extends both parts, and adding $\#$ makes the runtime part
empty.

\subparagraph{Models in Grothendieck toposes} \label{ssub:mod-gro}

The construction of $\cl S$ can be generalised beyond
$\s{Fam}(\Set)$, to an arbitrary Grothendieck topos $\re G$ that supports
squashed universes. Gratzer, Shulman and Sterling \cite{Gratzer2022-gn} have
shown (classically) that any Grothendieck topos $\re G$ admits a lifting of
$\{\Set_\ell\}_{\ell < \omega}$ where each $\univ_\ell$ contains all
$\ell$-small type families; Streicher~\cite{Streicher2005-tm} showed this
constructively for presheaf toposes. So all Grothendieck toposes support
universes. Which support squashed universes? We have written
$\amazing{\phi}{\univ}$ to imply that $\amazing{\phi}{-}$ is a functor. For
$\s{Fam}(\Set)$, it takes an object $X \triangleq (X_0, X_1)$ to $((x : X_0)
\times X_1\ x,\ \lambda\_.\std 1)$. It is uniquely characterised by the fact
that it is the \emph{right adjoint} to exponentiation by $\phi$:
$$
(\phi \to -) \dashv \amazing{\phi}{-}
$$
This construction has been studied in the context of type theory before
\cite{Licata2018-fe,Nuyts2020-pu,Sterling2021-kp,Riley2024-zk}, frequently
denoted by the square root $\sqrt{-}$ symbol. When a proposition $\phi$ has a
right adjoint, it is called \emph{tiny}. Therefore, if $\re G$ has a tiny
proposition $\phi$, it supports squashed universes with respect to $\phi$. This
has been observed by Sterling \cite{Sterling2023-jx} for essentially the same
purpose. This is not an overly restrictive condition either: in a presheaf topos
$\Psh{C}$, a representable $\yoneda X$ is tiny whenever $C$ has products
with $X$. As a result, any Grothendieck topos with a tiny proposition supports a
model of \ZTT. We leave spelling out the details of the general construction for
future work. This, along with the satisfaction of the realignment axiom
\cite{Gratzer2022-gn}, would justify the use of synthetic Tait computability
\cite{Sterling2022-ym} for the metatheory of \ZTT.

\section{Code extraction}\label{sub:code-extr}

The main purpose of \ZTT is to provide a practical language for programming with
dependent types, so it is useful to be able to extract executable code from \ZTT
programs. In particular, the code extraction process should erase all erased
terms, and preserve the computational behaviour of runtime terms. In this
section, we show how to extract code from $\ZTT$ programs by interpreting into a
\emph{presheaf model} of \ZTT over the untyped lambda calculus.

\begin{definition}[Untyped lambda calculus]\label{def:code-extr}
The untyped lambda calculus $\lambda$ quotiented by $\beta\eta$-equality is a SOGAT given by
\cref{eq:lc-sogat}. Its initial GAT model is the CwF $\re 0_\lambda$, where contexts
are natural numbers, substitutions $\re 0_\lambda.\Sub\ n\ m$ are $m$-tuples of lambda
terms with $n$ free variables, types are trivial (a single type $\star$), and terms
$\re 0_\lambda.\Tm\ n$ are untyped lambda terms with $n$ free variables \cite{Castellan2019-sh}.
This CwF supports $\Pi$ types by lambda abstraction and
application, and (positive) $\Sigma$ types, natural numbers, and other data types by
Church encoding.
\end{definition}

\begin{definition}[Code extraction model of \ZTT{} \agdaicon{https://cthe.me/redirects/erasure/extraction}]
  The code extraction model $\cl{E}$ interprets \ZTT into the
  base category of presheaves over the syntax of the untyped lambda calculus,
  $\Psh{\re 0_\lambda}$, meaning it is a $(\Sigma,\Pi_{\s R})$-CwF morphism
  $\cl E : \ZTT \to \re{Psh}_{\omega + 1}(\Psh{\re 0_\lambda})$ where:
  \[
    \begin{array}{c@{\qquad\qquad}c}
      \begin{array}{l@{\ }l}
        \cl E.\Ty_\ell &\triangleq \std{1} \\
        \cl E.\# &\triangleq \std{0}
      \end{array}
      &
      \begin{array}{l@{\ }l}
        \cl E.\Tm_0\ A &\triangleq \std{1} \\
        \cl E.\Tm_\omega\ A &\triangleq \yoneda {\re 0_\lambda}.\Tm
      \end{array}
    \end{array}
  \]
We interpret the \emph{runtime} (mode $\omega$) $\Pi$ and $\Sigma$
as the corresponding Church-encoded untyped structures in $\Psh {\re 0_\lambda}$, while the \emph{erased}
ones disappear. For example, we have
\[
\cl E.\Pi_0\ A\ B \triangleq \star \quad \text{(all types are trivial)} \qquad
\cl E.\s{lam}_0\ t \triangleq t \, \star \qquad \cl E.\s{app}_0\ f\ a \triangleq f
\]
The erasure coercions $\uptm$ and $\downtm$ reduce to the \emph{ex falso quodlibet} principle and
the terminal map respectively.
Universes disappear as well, since they only exist in the erased fragment.
\end{definition}

The double presheaf codomain is needed because we map the representable sort
$\#$ to $\mathbb{0}$ which is not representable in $\Psh{\re 0_\lambda}$ (there
is no $\re 0_\lambda$-context $n$ such that $\yoneda n \simeq \std 0$).\footnote{Alternatively, we could form a
\emph{higher-order model}, followed by contextualisation \cite{Bocquet2023-bu}
to get a GAT model in $\Psh{\re 0 _\lambda}$.} Upon unfolding the GAT model
$\widetilde{\cl E} : \ZTT^{\s{fo}} \to \Set$ corresponding to the
above, we can compute that contexts are ${\re 0_\lambda}$-presheaves, where
context extensions are interpreted as:
\[\begin{array}{c@{\qquad}c}
  \begin{array}{l@{\ }l}
    \llb{\bullet}_{\cl E}\ n & = \std{1}  \\
    \llb{\Gamma \rhd_\omega A}_{\cl E}\ n & = \llb\Gamma_{\cl E}\ n \times {\re 0_\lambda}.\Tm\ n
  \end{array}
  &
  \begin{array}{l@{\ }l}
    \llb{\Gamma \rhd_0 A}_{\cl E}\ n & = \llb\Gamma_{\cl E}\ n \times \std{1}\ (\simeq \llb\Gamma_{\cl E}\ n) \\
    \llb{\Gamma \rhd \#}_{\cl E}\ n & = \llb\Gamma_{\cl E}\ n \times \std{0}\ (\simeq \std 0)
  \end{array}
\end{array}\]
Up to isomorphism, adding an erased variable does nothing, adding a runtime
variable adds a lambda term, and adding an erasure marker makes the context
uninhabited.
  
To extract a program from a closed \ZTT term $t : \re 0_{\ZTT}.\Tm_\omega\ {\bullet}\ A$,
we interpret it in the code extraction model to get a closed lambda term:
\[\begin{aligned}
  \llb{t}_{\cl E} :&\ \cl{E}.\Tm_{\omega}\ \llb{\bullet}_{\cl E}\ \llb{A}_{\cl E}
    = \s{Hom}_{\Psh {\re 0_\lambda}}(\std{1}, {\re 0_\lambda}.\Tm)
    \simeq \s{Hom}_{\Psh {\re 0_\lambda}}(\yoneda 0, {\re 0_\lambda}.\Tm)
    \simeq \re 0_\lambda.\Tm\ 0 \,.
\end{aligned}\]
To interpret open terms, we observe that syntactic contexts which do not contain an
erasure marker are representable in $\Psh{\re 0_\lambda}$.
\begin{lemma}
  If $\# \not\in \Gamma$ (in \ZTT syntax), then $\llb\Gamma_{\cl E}$
  is representable -- there is a natural number $c_\Gamma$ which
  counts the runtime bindings in $\Gamma$, satisfying
  \[
    \llb\Gamma_{\cl E} \simeq \yoneda c_\Gamma \,.
  \]
  \begin{proof}
    By induction on contexts $\Gamma : \re 0_{\ZTT}.\Con$.
  \end{proof}
\end{lemma}
\begin{corollary}
  For $\# \not\in \Gamma$, there is an extraction map $|-| : \re 0_{\ZTT}.\Tm_\omega\ \Gamma\ A 
   \to \re 0_{\lambda}.\Tm\ c_\Gamma$.
  \begin{proof}
  \[\begin{aligned}
    \llb{t}_{\cl E} :\ \cl{E}.\Tm_{\omega}\ \llb\Gamma_{\cl E}\ \llb A_{\cl E}
      &= \s{Hom}_{\Psh{\re 0_\lambda}}(\llb\Gamma_{\cl E}, {\re 0_\lambda}.\Tm) \\
      &\simeq \s{Hom}_{\Psh{\re 0_\lambda}}(\yoneda c_\Gamma, {\re 0_\lambda}.\Tm) \\
      &\simeq {\re 0_\lambda}.\Tm\ c_\Gamma \,.
  \end{aligned}\]
  \end{proof}
\end{corollary}

This is the `purest' code extraction model we can formulate; in practice, we
would choose a richer untyped target that includes primitives for pairs and
inductive types (and one that would support full negative pairs), but its construction
would be entirely analogous.

\subsection{Correctness of code extraction}\label{sub:code-extr-corr}

To show that code extraction preserves the computational behaviour of \ZTT
programs, we set up a logical relation between the code extraction model $\cl E$
(\cref{def:code-extr}) and the logical interpretation in $\Set$. We do so by
building a model $\s{Gl}(F)$ of \ZTT extended with natural numbers $\s{Nat}$ in
a glued category. This model is given by gluing along a morphism of
$\ZTT$-models $F : \re 0_{\ZTT} \to \cl S$. It is the result of combining two
other morphisms. The \emph{first} morphism is the composite
\[
\begin{tikzcd}
  \re 0_{\ZTT} \ar[r, "{{\llb-_{\cl E}}}"] & \cl E \ar[rr, "P \mapsto P\,0"]
  & & \cl S_{\Set}
\end{tikzcd}
\]
which evaluates a $\re 0_\lambda$-presheaf produced by code
extraction at the empty context $0$ (in other words, the global sections pseudo-morphism \cite{Kaposi2019-lp}).
Its target is $\cl S_\Set$, the standard model of \ZTT valued in \Set rather than $\s{Fam}(\Set)$ with $\#$
interpreted by the singleton set $\Unit$.
The \emph{second} morphism is the interpretation morphism $\llb{-}_{\cl S_{\Set}} : \re 0_{\ZTT}
\to \cl S_\Set$ itself. These two define the morphism $F$,
valued in the standard $\s{Fam}(\Set)$ model $\cl S$, via:
\[
\begin{tikzcd}
  \re 0_{\ZTT} \ar[rrrr, "{\Gamma \mapsto (\llb{\Gamma}_{\cl S_\Set}, \lambda \_.\ \llb{\Gamma}_{\cl E} 0)}"] & & & & \cl S
\end{tikzcd}
\]
sending a context $\Gamma$ to the family $F \Gamma \triangleq (\llb{\Gamma}_{\cl S_\Set}, \lambda
\_.\ \llb{\Gamma}_{\cl E}\ 0)$. The standard $\Set$ interpretation is at
the base of each object, and the code extraction is at the fibers. From $F$,
we can construct a \emph{displayed} $\ZTT$ model $\s{Gl}(F)$ analogously
to the construction of Kaposi et~al. \cite{Kaposi2019-lp}.

The underlying category of $\s{Gl}(F)$ is $\s{Fam}(\Set) \downarrow F$
(where $F$ really means the \emph{underlying functor} between categories of
contexts). Each context in $\s{Gl}(F)$ consists of:
\begin{itemize}
\item a syntactic context $\Gamma : \re 0_{\ZTT}.\Con$
\item a `base' predicate $\Gamma^{\approx}_0 : \llb{\Gamma}_{\cl S_\Set} \to \Set_{\omega}$
\item a `fiber' predicate $\Gamma^{\approx}_1 : (\gamma_{\cl S} : \llb{\Gamma}_{\cl S_\Set}) \to \Gamma^{\approx}_0\ \gamma_{\cl S} \to \llb{\Gamma}_{\cl E}\ 0 \to \Set_{\omega}$
\end{itemize}
This comes with an evident projection morphism $(\Gamma, \Gamma^{\approx}_0, \Gamma^{\approx}_1)
\mapsto \Gamma$ into the syntax $\re 0_{\ZTT}$, which by initiality has a
section. This can be approximately thought of as a single binary relation
that relates the set interpretation with code extraction. The nuance is that,
in order to handle universes correctly, which are erased but must still carry
logical predicate data, we need a base predicate $\Gamma^{\approx}_0$ which exists even in
erased contexts to store them.

The sorts of the displayed model correspond to the induction
motives of the logical relation, which are presented in \cref{fig:logrel}.
\begin{figure}[h]
{$$\begin{array}{ll}
  (\Gamma : \s{Con})^{\approx} &:
    (\Gamma^{\approx}_0 : \llb\Gamma_{\cl S_\Set} \to \Set_{\omega}) \times (\Gamma^{\approx}_1 : \forall \gamma_{\cl S}.\ \llb\Gamma_{\cl E}\ 0
    \to \Gamma^{\approx}_0\ \gamma_{\cl S} \to \Set_{\omega}) \\[0.5em]
  (\sigma : \s{Sub}\ \Gamma\ \Delta)^{\approx} &:
    (\sigma^{\approx}_0 : \forall \gamma_{\cl S}.\ \Gamma^{\approx}_0\ \gamma_{\cl S} \to \Delta^{\approx}_0\ (\llb\sigma_{\cl S_\Set}\ \gamma_{\cl S})) \\
    &\qquad \times\ (\sigma^{\approx}_1 : \forall \gamma_{\cl E},\gamma_0.\ \Gamma^{\approx}_1\ \gamma_{\cl E}\ \gamma_0
    \to \Delta^{\approx}_1\ (\llb\sigma_{\cl E}\ 0\ \gamma_{\cl E})\ (\sigma^{\approx}_0\ \gamma_0)) \\[0.5em]
  (A : \s{Ty}_\ell\ \Gamma)^{\approx} &:
    (A^{\approx}_0 : \forall \gamma_{\cl S}.\ \Gamma^{\approx}_0\ \gamma_{\cl S} \to \llb A_{\cl S_\Set}\ \gamma_{\cl S} \to \Set_{\ell}) \\
    &\qquad \times\ (A^{\approx}_1 : \forall \gamma_0,a_{\cl S}.\ \re 0_\lambda.\s{Tm}\ {0}
    \to A^{\approx}_0\ \gamma_0\ a_{\cl S} \to \Set_{\ell}) \\[0.5em]
  (a : \s{Tm}_0\ \Gamma\ A)^{\approx} &:
    \forall \gamma_{\cl S}.\ (\gamma_0 : \Gamma^{\approx}_0\ \gamma_{\cl S}) \to A^{\approx}_0\ \gamma_0\ (\llb a_{\cl S_\Set}\ \gamma_{\cl S}) \\[0.5em]
  (a : \s{Tm}_\omega\ \Gamma\ A)^{\approx} &:
    (a^{\approx}_0 : \forall \gamma_{\cl S}.\ (\gamma_0 : \Gamma^{\approx}_0\ \gamma_{\cl S}) \to A^{\approx}_0\ \gamma_0\ (\llb a_{\cl S_\Set}\ \gamma_{\cl S})) \\
    &\qquad \times\ (a^{\approx}_1 : \forall \gamma_{\cl E},\gamma_0.\ \Gamma^{\approx}_1\ \gamma_{\cl E}\ \gamma_0
    \to A^{\approx}_1\ (\llb a_{\cl E}\ 0\ \gamma_{\cl E})\ (a^{\approx}_0\ \gamma_0)) \\[0.5em]
  (p : \#\in \Gamma)^{\approx} &: \llb\Gamma_{\cl E}\ 0 \to \std{0}
\end{array}$$}
\caption{The motives of the logical relation for code extraction correctness.}
\label{fig:logrel}
\end{figure}
This can be thought of as the logical relation version of the model described in
\cref{ssub:famset}; the interpretation of types is `squashed' because it is
indexed by erased contexts only, but packs both base and fiber data. The main
interesting component for our purposes is the type of natural numbers, where we
relate the two interpretations in the fiber: 
$$
\s{Nat}^{\approx} \triangleq (\lambda \gamma_0,n_{\s s}.\
\std{1},\ \lambda \{n_{\s s}\}, n_{\s e},\star.\ n_{\s e} = \un{succ}^{n_{\s
s}}\ \un{zero})
$$
where $a^n\ b \triangleq \s{rec}_{\std{N}}\ b\ a\ n$. We omit
the rest of the interpretation; see our Agda formalisation (\agdaicon{https://cthe.me/redirects/erasure/extraction-correctness}).\footnote{The formalisation works internally to $\Psh{\re 0_\lambda}$, which allows us
to directly use a second-order model of $\lambda$ rather than closed first-order terms. We include some
notes there about its relation to this version.}
From this model we obtain various useful correctness properties of extraction. Below,
we write $|-|$ for the code extraction of closed terms, $\llb{-}$ for the $\Set$ interpretation
of closed terms, and operate purely in the syntax $\re 0_{\ZTT}$.

\begin{theorem}[Canonicity \agdaicon{https://cthe.me/redirects/erasure/extraction-correctness}]\label{thm:canon}
  Every closed term $n : \Tm_\omega\ {\bullet}\ \s{Nat}$ is extracted to the numeral of
  its set interpretation: $|n| = \un{succ}^{\llbracket n \rrbracket}\ \un{zero}$.
\end{theorem}

\begin{theorem}[Tracking \agdaicon{https://cthe.me/redirects/erasure/tracking}]\label{thm:track}
  The extraction of any syntactic function tracks its semantic
  interpretation -- for all $f : \Tm_\omega\ \bullet\ (\Pi\ \s{Nat}\ \s{Nat})$ and all $k : \std{N}$,
  $\un{app}\ |f|\ (\un{succ}^k\ \un{zero}) = \un{succ}^{\llb f \ k}\ \un{zero}$.
\end{theorem}

\begin{theorem}[Non-interference \agdaicon{https://cthe.me/redirects/erasure/non-interference}]\label{thm:non-inter}
  Every erased function is constant at runtime -- for all $f : \Tm_\omega\ \bullet\ (\Pi_0\ \s{Nat}\ \s{Nat})$, there exists a $k : \std{N}$ such that for all $x : \Tm_0\ \bullet\ \s{Nat}$,
  $|\s{app}_0\ f\ x | = | f | = \un{succ}^k\ \un{zero}$.
  Moreover, $\llbracket f \rrbracket$ is the constant function returning $k$.
\end{theorem}

\section{Implementation} \label{sec:implementation}

We have implemented a toy elaborator for type theory
with erasure. This is based on Andr\'as Kov\'acs' \texttt{elaboration-zoo} which
contains toy implementations of dependent type theory. Our implementation is a
modification of the implementation of implicit arguments and metavariables
\cite{KovacsUnknown-sx}. The surface language is essentially the
language presented in \cref{sec:informal-itt}, with mode-aware $\Pi$ types and a
single universe $\s U : \s U$. The coercions $\uptm$/$\downtm$ and the marker
$\#$ are inserted automatically; the user never interacts with them.

In the repository, we include two variants of the elaboration algorithm: one
which \emph{inserts} coercions during elaboration, and one which keeps coercions
implicit. The latter is closer to existing implementations of
erasure. We keep the former as a proof of concept that it is possible to have a
structural phase distinction which can be elaborated from a `substructural'
source language. In both cases, we make a simplification to the representation
of contexts: the theory \ZTT as presented features context extensions by $\#$
which can end up anywhere in the context, and can appear multiple times.
However, $\#$ is a proposition, so it doesn't matter which particular witness we
use. It is therefore sufficient to keep a boolean flag indicating the mere
\emph{presence} of $\#$ in a context otherwise containing only 0/$\omega$
bindings. This simplifies the handling of variables, since we don't need to
offset de~Bruijn indices/levels by $\#$.

The predominant source of complexity in elaborating such languages is the
pattern unification algorithm that solves metavariables. Although there are
theoretical foundations of pattern unification for type theory
\cite{Abel2011-sx}, such a formal analysis has not been performed for languages
with erasure or other modalities. The implementations of pattern unification in
the wild are thus `engineering efforts', which can go wrong. Despite checking
quantities only \emph{after} the whole program has been elaborated, Idris 2
sometimes solves metas in a weaker quantity than required, which can lead to
undefined behaviour at runtime \cite{Dunham2024-up}. On the other hand, Agda
sometimes fails to detect unsolvable metas in the presence of erasure
\cite{Cockx2021-mq}. 

Luckily, because our implementation is based on a structural core, we can
directly reuse the theory of pattern unification to obtain a correct
implementation. Our unification algorithm does not require a separate mode check
after typechecking as opposed to Idris. It also does not need a dedicated
generalisation mechanism for promoting erased metavariables to runtime
metavariables as opposed to Agda. In the repository, we include some test cases
that otherwise fail in these languages due to the incompleteness of their
respective mechanisms (they are the test cases in the issues linked above). We
have justified our formulation of pattern unification by some semi-formal notes
in the same repository. The core idea is that the process of renaming and
performing occurrence checking on candidate solutions handles not only regular
variables, but also witnesses of the erasure marker $\#$ (which are morally also
just variables). As a result, we only need one kind of metavariable (runtime)
and generalisation is no longer necessary.

\section{Related work}\label{sec:related}

Our approach to erasure is based on \emph{synthetic phase distinctions},
pioneered by Sterling and Harper \cite{Sterling2021-pm}, whose roots go back to
the work on phase distinctions of Cardelli \cite{Cardelli1988-zp} and Harper,
Mitchell, and Moggi \cite{Harper1990-el}. Cardelli's work is the closest to erasure in terms
of purpose, but is formulated as an `indexed' type theory. In his thesis
\cite{Sterling2022-ym}, Sterling develops a theory of synthetic phase
distinctions as a tool for constructing logical relations (synthetic Tait
computability), but with various applications to programming languages. Most
recently, Grodin et~al.~\cite{Grodin2025-cl} have showcased the possible use
cases of a language with phase distinctions. In such settings one has access to
open as well as closed modalities, corresponding to open and closed subtoposes
of the topos in which the language has semantics. Using this terminology, our
formulation of erasure is the open modality for the proposition $\#$. The
closest work along these lines to ours is in the form of a blog post by Sterling
\cite{Sterling2023-jx}, which contains some ideas about how synthetic Tait
computability relates to QTT; in particular, he observes the need for squashed
universes.

Erasure in dependent types was explored by Mishra-Linger
et~al.~\cite{Mishra-Linger2008-zy} in the context of \emph{pure type systems}.
With the work of Gundry and McBride \cite{Gundry2013-px} as precursor, the
modern approach to erasure has been QTT by McBride \cite{McBride2016-oa} and
Atkey \cite{Atkey2018-pj}. We intend to characterise the relationship between
our theory and QTT, whose models are \emph{quantitative} CwFs (QCwFs), in the
future. For now, we make the observation that a \emph{structural} QCwF (meaning
with the semiring $\cl R = \{0 < \omega\}$) is an indexed CwF $U : \cl L \to \cl
C$. Gluing along the reindexing functor $U^\star : \Psh{\cl C} \to \Psh{\cl
L}$ yields a presheaf category with a tiny proposition that models \ZTT. In the
other direction, given a \ZTT-model $\cl M$, strictifying the pseudo-morphism
$\cl M \to \cl M / \#$ into the slice CwF $\cl M / \#$ yields a
structural QCwF.

More recently, Danielsson has explored some constructions in type theory with
erasure \cite{Danielsson_undated-pj}, and Abel et~al.~have explored its
integration with cubical type theory for Cubical Agda \cite{AbelUnknown-od}. We
expect that our system can extend the SOGAT of cubical type theory
\cite[4.6.3]{Uemura2021-jq} with a mode split for terms, and thus replicate
Abel's system structurally. Favier \cite{FavierUnknown-wx} has shown that
erasure behaves like an open modality in Agda, showing a synthetic Artin
fracture theorem. Besides this, there has also been work on theories with `mode
splits', notably type theory with colours \cite{Jean-Philippe2013-js}. This
style of system, where terms at each mode need not be the same, can be
replicated in our system; we are free to add equations that apply only under the
$\#$ marker, collapsing data in the erased phase that otherwise exists at the
runtime phase. 

\section{Conclusion}\label{sec:concl}

We have developed a fully structural theory of erasure using the formalism of
SOGATs, and explored various syntactic and semantic models.  
In the future, we would like to explore more extensions to the theory. The most
immediate is to support runtime types. This is relatively straightforward to add
and simplifies the semantics by avoiding the need for squashed universes; we are
mostly interested in exploring the utility of this feature for programming.
Another extension is to add more phase distinctions. Suppose we allow equations
like $\beta$-reduction only under $\#$. We could then add a second
\emph{disjoint} phase distinction $\text{\textdollar}$ to internalise the \emph{code
extraction} morphism. Doing so would allow us to control and reason about
compilation output internally to the language, for example to specify runtime
optimisations.

\bibliography{bib}

\end{document}